\documentclass[12pt,a4paper,final]{iopart}
\usepackage{iopams}  
\usepackage{graphicx}
\usepackage[breaklinks=true,colorlinks=true,linkcolor=blue,urlcolor=blue,citecolor=blue]{hyperref}

\expandafter\let\csname equation*\endcsname\relax
\expandafter\let\csname endequation*\endcsname\relax

\usepackage{
amsthm,
amsfonts,
mathtools
}

\DeclareMathOperator{\sinc}{\text{sinc}}
\DeclareMathOperator{\Ei}{\text{Ei}}

\DeclareMathOperator{\erf}{\text{erf}}

\newtheorem{theorem}{Theorem}[section]
\newtheorem{lemma}[theorem]{Lemma}
\newtheorem{proposition}[theorem]{Proposition}

\usepackage{CJKutf8}

\begin{document}

\title[]{Integration by differentiation: new proofs, methods and examples}

\begin{CJK*}{UTF8}{gbsn}

\author{Ding Jia (贾丁)$^{1,2}$, Eugene Tang$^3$,  Achim Kempf$^{1,2}$}
\address{$ $ \\ $^1$Departments of Applied Mathematics\\
University of Waterloo, Ontario N2L 3G1, Canada\\
${}^2$ Perimeter Institute for Theoretical Physics\\ Waterloo, ON, N2L 2Y5, Canada\\
$^3$Division of Physics, Mathematics and Astronomy
\\California Institute of Technology, Pasadena, CA 91125, United States}
\ead{d7jia@edu.uwaterloo.ca}

\end{CJK*}

\begin{abstract}
Recently, new methods were introduced which allow one to solve ordinary integrals by performing only derivatives.
These studies were originally motivated by the difficulties of the quantum field theoretic path integral, and correspondingly, the results were derived by heuristic methods. Here, we give rigorous proofs for the methods to hold on fully specified function spaces. We then illustrate the efficacy of the new methods by applying them to the study of the surprising behavior of so-called Borwein integrals.  
\end{abstract}


\section{Introduction}
Integration is generally harder to perform than differentiation. In particular, functional integrals, i.e., Feynman path integrals, are harder to work with, or even to define, than functional derivatives. See, for example, \cite{weinberg1996quantum,ticciati1999quantum,fujikawa1979path,davies1990heat,kempf2009information}. In this general context, methods have recently been found that allow one to integrate, or path integrate, by performing only derivatives, or functional derivatives, respectively.   \cite{kempf2014new,kempf2015path}. 
While these methods were originally introduced to express quantum field theoretical path integrals in terms of much easier-to-handle functional derivatives, it was found that the new integration-by-differentiation techniques also add valuable tools to the toolbox for the evaluation of ordinary integrals.

So far, however, the domain of validity of the new methods has been little explored. Therefore, we here give rigorous proofs which show that the new methods hold for certain large classes of  function spaces. Further, we illustrate the power of the new methods by applying them to sequences of so-called Borwein integrals. These integrals are nontrivial to evaluate and are known to exhibit rather curious behaviors. We demonstrate how the new methods allow one to evaluate even the normally complicated Borwein integrals quickly and transparently. 

\section{Integration by differentiation}
Assume that $f:\mathbb{R}\rightarrow \mathbb{R}$ is a function which has a power series expansion, $f(x)=\sum_{k=0}^\infty a_kx^k$, whose radius of convergence is nonzero. For any choice of $r\in\mathbb{C}$, we define a differential operator which we will denote by $f(r\partial_x)$. By definition, $f(r\partial_x)$ acts on smooth functions $\varphi(x)$ as
\begin{equation}f(r\partial_x)\varphi(x) := \lim_{N\rightarrow\infty}\sum_{k=0}^Na_k(r\partial_x)^k\varphi(x),\label{eq:series-def}\end{equation}
whenever the latter limit is convergent. For convenience, we defined  $\partial_x:=d/dx$. For example, if $f(x)= \exp(x)$ then $f(r\partial_x) = \exp(r\partial_x)$, which can act on smooth functions such as $\varphi(x)=\sin(x)$, giving
\begin{equation}f(r\partial_x)\varphi(x) 
= e^{r\partial_x}\sin(x)=
\lim_{N\rightarrow\infty}\sum_{k=0}^N\frac{(r\partial_x)^k}{k!}\sin(x)
 = \sin(x+r).
\end{equation}
The last step follows from the Taylor expansion. \medskip\newline Using this definition of $f(r\partial_x)$, the following integration by differentiation methods have been introduced \cite{kempf2014new,kempf2015path}:
	\begin{eqnarray}
	\int_a^b f(x)\ dx & = & \lim_{y\rightarrow 0}f(-i\partial_y)\frac{e^{iby}-e^{iay}}{iy}, \label{eq:abi0}\\
	\int_a^b f(x)\ dx &  =&  \lim_{y\rightarrow 0}f(\partial_y)\frac{e^{by}-e^{ay}}{y}, \label{eq:ab0}
	\\\nonumber
	\\\nonumber
	\\
	\int_0^\infty f(x)~dx & = & \lim_{y\rightarrow 0^+} f(-\partial_y)~\frac{1}{y},\label{eq:+0}\\
	\int_{-\infty}^0 f(x)~dx & = & \lim_{y\rightarrow 0^+} f(\partial_y)~\frac{1}{y},\label{eq:-0}\\
	\int_{-\infty}^\infty f(x)~dx & = & \lim_{y\rightarrow 0^+} \bigg(f(\partial_y)+f(-\partial_y)\bigg)~\frac{1}{y},\label{eq:+-0}
	\\\nonumber\\\nonumber\\
	\int_{-\infty}^\infty f(x)~dx & = & \lim_{y\rightarrow 0^+}2\pi f(-i\partial_y)~\delta(y).\label{eq:delta0} 
	\end{eqnarray}
Further, these equations are the ``zero frequency'' cases of the following general formulas for the Fourier and Laplace transforms:
	\begin{eqnarray}
	\int_a^b f(x)e^{ixy}\ dx & = & f(-i\partial_y)\frac{e^{iby}-e^{iay}}{iy}, \label{eq:abi}\\
	\int_a^b f(x)e^{xy}\ dx & = & f(\partial_y)\frac{e^{by}-e^{ay}}{y}\label{ab},\label{eq:ab}
	\\\nonumber
	\\\nonumber
	\\
	\int_0^\infty f(x)e^{-xy}~dx & = & f(-\partial_y)~\frac{1}{y}\label{eq:+},\\
	\int_{-\infty}^0 f(x)e^{xy}~dx & = & f(\partial_y)~\frac{1}{y}\label{eq:-},
	\\\nonumber
	\\\nonumber
	\\
	\int_{-\infty}^\infty f(x)e^{ixy}~dx & = & 2\pi f(-i\partial_y)\delta(y).\label{eq:delta} 
	\end{eqnarray}
	\linebreak
	\noindent For example, here is a calculation based on equation (\ref{eq:ab0}):
	\begin{align}
	\int_a^b xe^{-x}\  dx=&\lim_{y\rightarrow 0} e^{-\partial_y} \partial_y \frac{e^{by}-e^{ay}}{y}
	\\
	=&\lim_{y\rightarrow 0} e^{-\partial_y} \frac{(by-1)e^{by}-(ay-1)e^{ay}}{y^2}
	\\
	=&\lim_{y\rightarrow 0}  \frac{(by-b-1)e^{b(y-1)}-(ay-a-1)e^{a(y-1)}}{(y-1)^2}
	\\
	=&-(b+1)e^{-b}+(a+1)e^{-a}.
	\end{align}
	We used here the fact that $e^{a\partial_y}f(y)=f(y+a)$, i.e., that $e^{a\partial_y}$ acts as a translation operator (see Lemma \ref{lem:trans}).
	While the integration by differentiation formulas were originally derived heuristically, our aim here is to prove them rigorously and establish conditions under which they apply.

\section{Main Results}\label{sec:main}

In this section we present a list of propositions that put the above integration by differentiation methods on a rigorous footing. The proofs are given in Section \ref{sec:proofs}. 
Throughout, we will continue to define $f(r\partial_x)$ through the power series expansion of $f$, unless otherwise mentioned.

Indeed, before turning to the main results, let us briefly discuss an alternative way in which we could define the function of a derivative, namely not through a power series expansion but instead through the spectral calculus. Concretely, using the spectral theorem, we could define 
\begin{equation}f(-i\partial_x) \equiv \mathcal{F}M_{f(-x)}\mathcal{F}^{-1},\label{eq:spectral-def}\end{equation}
where $\mathcal{F}$ is the Fourier transform, which we define through:
\begin{equation}\mathcal{F}[f](y) = \hat{f}(y) = \frac{1}{\sqrt{2\pi}}\int_\mathbb{R} f(x)e^{ixy}\ dx,\end{equation}
Here, $M_{f(-x)}$ is the multiplication operator $M_{f(-x)}\varphi(x) = f(-x)\varphi(x)$. Note the unconventional definition of the Fourier transform here, with kernel $e^{ixy}$ instead of $e^{-ixy}$. This is simply to minimize unnecessary factors of $-1$ later on.

Note that equation (\ref{eq:spectral-def}) is precisely the definition of $f(-i\partial_x)$ as a pseudo-differential operator. This definition, while tempting, does not work for the purposes of this paper. Our purpose for defining $f(-i\partial_x)$ here is to develop useful new methods for performing the Fourier transform (as well as the associated integrals which arise as the zero frequency limit of the Fourier transform). If we were to define the function of a derivative as a pseudo-differential operator, we would have to assume the ability to perform the Fourier transform from the very start, which would defeat the intended purpose of the methods. It is therefore important for our purposes that we do not define functions of derivatives using the theory of pseudo-differential operators. Instead, we define the functions of derivatives through their power series expansions. In this way, the new methods merely assume the ability to explicitly perform derivatives. 

Instead of our methods being based upon the theory of pseudo-differential operators, our methods originate from the standard practice for evaluating quantum field theoretic path integrals by deriving the Feynman rules. There, the interaction term in the action is viewed as a function of derivatives (with respect to source fields) in the same way as we do here, namely through that function's Taylor expansion. Our methods contain this example as a special case, and it has been shown that these generalized methods can be applied not only to path integration but also to regular integration and integral transforms \cite{kempf2014new,kempf2015path}. In the present paper we show that the new methods are reliable, namely by providing the first rigorously-proven results on sufficient conditions for the new methods to apply.

\subsection{Finite intervals}
\noindent We begin with the equations that are the most straightforward to establish rigorously, namely the formulas which involve only integrals over finite intervals. The basic results are summarized in the following proposition:

	\begin{proposition}\label{prop:finint}
		Suppose that $f(x)$ has a power series expansion $f(x)=\sum_{k=0}^{\infty}a_kx^k$ with a radius of convergence covering $[a,b]$.  Then equations (\ref{eq:abi0}), (\ref{eq:ab0}), (\ref{eq:abi}), and (\ref{eq:ab}) hold. 
\end{proposition}

\subsection{Laplace-type methods}
	\noindent From the previous finite integral equations, we may take the limits of integration to $\pm\infty$. As we shall show in section \ref{sec:proofs}, careful taking of the limit yields rigorous versions of (\ref{eq:+}) and (\ref{eq:-}): 
	
		\begin{proposition}\label{prop:intlap}
		If $f:\mathbb{R}\rightarrow \mathbb{R}$ is entire and Laplace transformable on $\mathbb{R}_+$, then
		\begin{eqnarray}
		\int_0^\infty f(x)e^{-xy}\ dx & = & \lim_{a\rightarrow\infty}f(-\partial_y)\frac{1-e^{-ay}}{y},\label{eq:intlap+}
		\end{eqnarray}
		for all $y\in \mathbb{R}_+$  for which the integral is convergent. Likewise, if $f$ is entire and Laplace transformable\footnote{By Laplace transformable on $\mathbb{R}_-$, we mean that there exists some $y_0\in\mathbb{R}_+$ such that $\int_{-\infty}^0 f(x)e^{xy}\ dt$ is convergent for all $y>y_0$.} on $\mathbb{R}_-$, then
		\begin{eqnarray}
		\int_{-\infty}^0 f(x)e^{xy}\ dx & = & \lim_{a\rightarrow \infty}f(\partial_y)\frac{1-e^{-ay}}{y}\label{eq:intlap-}
		\end{eqnarray}
		for all $y\in \mathbb{R}_+$ for which the integral is convergent.
	\end{proposition}
	
\noindent If the function $f$ is entire and integrable on the real or half-lines, then we may first take the zero frequency limit $y\rightarrow 0$ before taking the limits of integration to $\pm\infty$. This yields the following proposition: 

\begin{proposition}\label{prop:intzerofreq}
		If $f$ is entire and integrable on the half-lines $\mathbb{R}_+$ or $\mathbb{R}_-$, then we have
		\begin{eqnarray}
		\int_0^\infty f(x)\ dx & = & \lim_{a\rightarrow\infty}\lim_{y\rightarrow 0}f(-\partial_y)\frac{1-e^{-ay}}{y},\label{eq:intlap+0}\\
		\int_{-\infty}^0 f(x)\ dx & = & \lim_{a\rightarrow \infty}\lim_{y\rightarrow 0}f(\partial_y)\frac{1-e^{-ay}}{y}\label{eq:intlap-0},
		\end{eqnarray}
		respectively. 
		Then, if $f:\mathbb{R}\rightarrow \mathbb{R}$ is entire and integrable on $\mathbb{R}$, we have
		\begin{eqnarray}
		\int_{-\infty}^\infty f(x)~dx & = & \lim_{a\rightarrow\infty}\lim_{y\rightarrow 0} \left[f(\partial_y)+f(-\partial_y)\right]\frac{1-e^{-ay}}{y}\label{eq:intlap+-0},\\
		\int_{-\infty}^\infty f(x)\ dx & = & \lim_{a\rightarrow\infty}\lim_{y\rightarrow 0} f(-i\partial_y)\,2a\sinc(ay).\label{eq:intfourpre0}
		\end{eqnarray}

		\end{proposition}
Let us now turn our attention to equations (\ref{eq:+}) and (\ref{eq:-}). Looking at Proposition \ref{prop:intlap}, we notice a remarkable similarity between equations (\ref{eq:+}) - (\ref{eq:-}), and equations (\ref{eq:intlap+}) - (\ref{eq:intlap-}). Indeed, if we could exchange the limit $a\rightarrow \infty$ and the operator $f(-\partial_y)$ in equations (\ref{eq:intlap+}) and (\ref{eq:intlap-}), then we would immediately obtain (\ref{eq:+}) and (\ref{eq:-}). However, these operations in general cannot be exchanged, and so there is some subtlety in the sense in which equations (\ref{eq:+}) and (\ref{eq:-}) hold: 

From the right-hand side of (\ref{eq:+}), we can see that the power series differential operator acting on $1/y$ gives us a Laurent series in $y$ instead of the usual Taylor series. This is natural, considering that the convergence of the Laplace transform becomes better with increasing $y$. Given the fact that equation (\ref{eq:+}) holds with a Laurent series, it can happen that the domain of convergence for equation (\ref{eq:+}) is not the full domain of convergence of the corresponding Laplace transform.

For example, consider the application of equation (\ref{eq:+}) to the function $f(x)=e^{-x}$. This function is Laplace transformable, and the Laplace transform
\begin{equation}\mathcal{L}[e^{-x}](y) = \int_0^\infty e^{-x}e^{-xy}\ dx = \frac{1}{y+1}\end{equation}
is convergent for all $y > -1$. If we evaluate the same example with equation (\ref{eq:+}), then we find
\begin{equation}e^{\partial_y}\left(\frac{1}{y}\right) = \sum_{k=0}^\infty \frac{1}{k!}\frac{d^k}{dy^k}\left(\frac{1}{y}\right) = \sum_{k=0}^\infty (-1)^k\frac{1}{y^{k+1}}.\end{equation}
This is a proper Laurent series which equals $1/(y+1)$, as it should, but which is only convergent for $|y|>1$. Thus we see that the two domains differ and that, therefore, these problems warrant further investigation.

The underlying issue here is the fact that, in equation (\ref{eq:+}), the operator $f(-\partial_y)$ is acting on $1/y$, which has a pole at $y=0$. According to our new Proposition \ref{prop:intlap}, the use of the equation  (\ref{eq:intlap+}), which is the regularized version of equation (\ref{eq:+}), must succeed in yielding the full solution. 

Let us verify this: Indeed, in contrast to $1/y$, which is singular, the function $(1-e^{-ay})/y$ is an entire function. Applying (\ref{eq:intlap+}) to  our example $f(x)=e^{-x}$ then gives us
\begin{equation}\lim_{a\rightarrow\infty}e^{\partial_y}\left(\frac{1-e^{-ay}}{y}\right) = \lim_{a\rightarrow\infty}\frac{1-e^{-a(y+1)}}{y+1} = \frac{1}{y+1},\end{equation}
where the final limit is convergent for $y>-1$, which means that we obtain the Laplace transform on its full domain. Note that because of the fact that $(1-e^{-ay})/y$ is entire, we were able to use $e^{\partial_y}$ as the translation operator (see Lemma \ref{lem:trans}). We were not able to do this for $1/y$. Moreover, for each finite value of $a$, the resulting function
\begin{equation} e^{\partial_y}\left(\frac{1-e^{-ay}}{y}\right) = \frac{1-e^{-a(y+1)}}{y+1}\end{equation}
is an entire function of $y$. Therefore, equation (\ref{eq:intlap+}) avoids the convergence problem by obtaining its results as the limit of entire functions, bypassing the need for a Laurent series expansion. The two equations will of course agree within their common domain of convergence, and one may take the view of equation (\ref{eq:intlap+}) as giving an analytic continuation of equation (\ref{eq:+}). 

Now for practical purposes, equation (\ref{eq:+}) is easier to use than equation (\ref{eq:intlap+}). The problem with (\ref{eq:+}), however, is that as we just saw, its domain of convergence is nontrivial. The following proposition provides a practical criterion with mild conditions on $f$ for the use of equation (\ref{eq:+}):

\begin{proposition}\label{prop:laplace}
Let $f:\mathbb{R}\rightarrow \mathbb{R}$ be an entire function given by
\begin{equation}f(x) = \sum_{k=0}^\infty a_kx^k.\end{equation}
Let us also define the function $\tilde{f}(x)$ by
\begin{equation}\tilde{f}(x) = \sum_{k=0}^\infty |a_k|x^k.\end{equation}
Suppose that $\tilde{f}$ is Laplace transformable, such that the integral
\begin{equation}\mathcal{L}[\tilde{f}](y) = \int_0^\infty \tilde{f}(x)e^{-xy}\ dx\end{equation} 
is convergent for all $y>y_0\ge 0$. Then we have 
\begin{equation}\mathcal{L}[f](y) = \int_0^\infty f(x)e^{-xy}\ dx = f(-\partial_y)\frac{1}{y},\end{equation}
which holds for all $y > y_0$. The statement for the Laplace transform on $(-\infty,0]$ holds analogously.
\end{proposition}

Applying this proposition to $f(x)=e^{-x}$, we have $\tilde{f}(x) = e^x$, and so the Laplace transform for $\tilde{f}(x)$ is convergent for all $y>y_0=1$. Indeed, this is precisely the domain of convergence we found for the Laurent series before.

Finally, let us turn towards the zero frequency case of equations (\ref{eq:+}) and (\ref{eq:-}), given by equations (\ref{eq:+0}) - (\ref{eq:+-0}). From our previous example, we see that the domain of convergence for $\tilde{f}(y)$ can easily be smaller than the positive half line, i.e., with $y_0 > 0$. 
In this case, $y=0$ is not a limit point for the domain of convergence, and therefore the limit $y\rightarrow 0$ to obtain an integration formula cannot be taken. Therefore, to be safe, the two equations (\ref{eq:+0}) - (\ref{eq:+-0}) must, in general, be used in their regularized forms given by equations (\ref{eq:intlap+}) - (\ref{eq:intlap-}) and (\ref{eq:intlap+-0}). 

Fortunately, as we will now show, for certain classes of functions, these regularizations are not needed and the integration by differentiation equations (\ref{eq:+0}) - (\ref{eq:+-0}) produce the correct answers even when used directly. This is the content of the next proposition.

\begin{proposition}\label{prop:intlapsim}
Suppose $f:\mathbb{R}\rightarrow \mathbb{R}$ is an entire and integrable function of the form
\begin{align}\label{eq:fsimple}
	f(x)=\sum_{j=1}^{N} c_j e^{-b_j x} x^{n_j},
\end{align}
with $b_i>0$, $c_j \in \mathbb{C}$, and $n_j\in\mathbb{Z}$. Then equations (\ref{eq:+0}) - (\ref{eq:+-0}) hold with the action of $f(\pm\partial_y)$ on the right-hand side given formally (instead of as a power series). Namely, we define $e^{b\partial_y}$ to act by translation
$$e^{b\partial_y}\varphi(y) = \varphi(y+b),$$ 
and $\partial_y^{n}$ to give the $n$th derivative ($n>0$) or anti-derivative ($n<0$). 

Moreover, when anti-differentiating the choice of representative will not matter, i.e. the equations are insensitive to integration constants.
\end{proposition}    

The key difference between this proposition and the previous results is that the operator $f(\partial_y)$ appearing here is not a power series operator. Instead, it is an operator that acts formally to implement translations, integrations, and differentiations. This gives rigorous justification to how equations (\ref{eq:+0}) - (\ref{eq:+-0}) are used in practice (see \cite{kempf2014new,kempf2015path}).

\subsection{Fourier-type methods}

Taking the limit of the finite interval equation (\ref{eq:abi}), we have the following proposition:

	\begin{proposition}\label{prop:intfour}
		If $f:\mathbb{R}\rightarrow \mathbb{R}$ is an entire and Fourier transformable function, then
		\begin{equation}
		\int_{-\infty}^\infty f(x)e^{ixy}\ dx = \lim_{a\rightarrow\infty} f(-i\partial_y)~2a\sinc(ay).\label{eq:intfourpre}\end{equation}
	\end{proposition}
	\noindent
	Notice that the function $\frac{a}{\pi}\sinc(ax)$ converges to the Dirac delta in the weak limit as $a\rightarrow\infty$, so writing equation (\ref{eq:intfourpre}) as
	\begin{equation}\int_{-\infty}^\infty f(x)e^{ixy}~dx = \lim_{a\rightarrow\infty} 2\pi f(-i\partial_y)~\frac{a}{\pi}\sinc(ay),\end{equation}
	we can see that equation (\ref{eq:intfourpre}) is a regularized representation of equation (\ref{eq:delta}).

	However, as we shall later see, (\ref{eq:delta}) holds in a more general sense. We start by noticing that due to the presence of the Dirac delta, equation (\ref{eq:delta}) must be treated inherently as a distributional identity. Whereas equations (\ref{eq:abi}) - (\ref{eq:-}) can be treated without passing through distribution theory, equation (\ref{eq:delta}) cannot even be stated without the presence of a distribution. While this yields slightly more complexity, it also means that equation (\ref{eq:delta}) can be generalized to hold in a much broader sense than the others. We now clarify the precise setting for equation (\ref{eq:delta}) by introducing a few relevant definitions. The full technical treatment of equation (\ref{eq:delta}) is presented in section \ref{sec:proofs}. 

Recall that a distribution is a continuous linear functional on the space of smooth, compactly supported test functions on $\mathbb{R}$. We denote the space of test functions by $\mathcal{D}$, and its continuous dual space, i.e., the space of distributions, by $\mathcal{D}'$. The action of a distribution $f \in \mathcal{D}'$ on a test function $\varphi \in \mathcal{D}$ is denoted by the pairing $\langle f, \varphi\rangle \in \mathbb{C}$, and a sequence of distributions $f_n \in \mathcal{D}'$ is said to converge to $f\in \mathcal{D}'$ if and only if we have
\begin{equation}\lim_{n\rightarrow\infty}\langle f_n, \varphi\rangle = \langle f,\varphi\rangle
\end{equation}
for all test functions $\varphi \in \mathcal{D}$.

On might expect that equation (\ref{eq:delta}) holds as an ordinary distribution, but it is not so simple. To ensure that Fourier transforms and power series are well-defined, we must introduce an alternate space of test functions. 

Let $\mathcal{Z}\equiv \mathcal{F}(\mathcal{D})$ denote the space of functions obtained from a smooth, compactly supported test function via the Fourier transform. We will say that a sequence of functions $F_n \in \mathcal{Z}$ converges to $F\in\mathcal{Z}$ if and only if the respective Fourier transforms $f_n \in \mathcal{D}$ converge to a corresponding $f \in \mathcal{D}$. We will call $\mathcal{Z}$ the space of Paley-Wiener (PW) test functions. The continuous dual space $\mathcal{Z}'$ will then be called the space of Paley-Wiener (PW) distributions. The spaces $\mathcal{Z}$ and $\mathcal{Z}'$ turn out to be the proper setting for equation (\ref{eq:delta}).

In what follows, let $L^1(\mathbb{R}) + L^2(\mathbb{R})$ denote the space of functions which can be written as $f_1+f_2$, where $f_1 \in L^1(\mathbb{R})$ and $f_2 \in L^2(\mathbb{R})$. The next proposition gives the precise meaning of equation (\ref{eq:delta}) as a Paley-Wiener distribution.

\begin{proposition}\label{prop:distdelta}
Let $f \in L^1(\mathbb{R}) + L^2(\mathbb{R})$ be an entire function with power series expansion
\begin{equation}f(x) = \sum_{k=0}^\infty a_kx^k.\end{equation}
Then the equation
\begin{equation}\hat{f}(y) = \sqrt{2\pi}f(-i\partial_y)\delta(y) \equiv \sqrt{2\pi}\sum_{k=0}^\infty a_k(-i\partial_y)^k\delta(y)\end{equation}
holds as an equality of Paley-Wiener distributions in the space $\mathcal{Z}'$. Explicitly, for any Paley-Wiener test function $\varphi \in \mathcal{Z}$, we have
\begin{equation} \int_\mathbb{R} \hat{f}(y)\varphi(y)\ dy = \left\langle f(-i\partial_y)\delta(y),\varphi(y)\right\rangle = \sum_{k=0}^\infty (-i)^ka_k\varphi^{(k)}(0).\end{equation}
\end{proposition}
\noindent
For certain choices of $f$, we may evaluate (\ref{eq:delta}) formally and lift the distributional identity to a proper equality of ordinary functions. The following theorem is the equivalent of Proposition \ref{prop:intlapsim} for equation (\ref{eq:delta}). 

\begin{proposition}\label{prop:funcdelta} 
Let $f \in L^1(\mathbb{R}) + L^2(\mathbb{R})$ be an entire function of the form
\begin{equation}f(x) = \sum_{n=1}^N\frac{c_ke^{ib_kx} + d_ke^{-ib_kx}}{x^{n_k}},\label{eq:fsimple2}\end{equation}
for constants $c_k,d_k\in\mathbb{C}$, $b_k \in \mathbb{R}$, and $n_k \in \mathbb{N}_+$. Then we have
\begin{equation}\hat{f}(y) = \sqrt{2\pi}f(-i\partial_y)\delta(y),\label{eq:fourier_delta}\end{equation}
which holds as an equality of functions. The right-hand side is given by the formal action of $f(-i\partial_y)$ on $\delta(y)$, where $e^{b\partial_y}$ acts via translation by $b$, and where $\partial_y^{-n}$ acts by formal anti-differentiation, giving us  
\begin{equation}\partial_y^{-n}\delta(y) = R_{n-1}(y) = \frac{y^{n-1}}{(n-1)!}\Theta(y),\label{eq:polyramp}\end{equation} 
where $R_{n-1}$ is the generalized ramp function which defines an $n$th anti-derivative of the Dirac delta. Here, $\Theta(y)$ denotes the Heaviside step function.
\end{proposition}
As with Proposition \ref{prop:intlapsim}, the operators $e^{b\partial_y}$ and $\partial^{-n}_y$ acts formally (instead of as a power series operator) via translation and integration. Also like Proposition \ref{prop:intlapsim}, the choice of representative for the anti-derivative does not matter. For simplicity, we have taken $\partial_y^{-n}\delta(y) = R_{n-1}(y)$, but the result holds equivalently with $\partial_y^{-n}\delta(y) = R_{n-1}(y) + p_{n-1}(y)$ for any polynomial $p_{n-1}(y)$ of degree less than or equal to $n-1$.
\bigskip\newline
Finally, if $f$ is integrable on the real line, we can evaluate $\hat{f}(y)$ at $y=0$ to conclude that
\begin{equation}\int_\mathbb{R}f(x)\ dx = \lim_{y\rightarrow 0} 2\pi f(-i\partial_y)\delta(y),\end{equation}
which is precisely equation (\ref{eq:delta0}), for functions of the form (\ref{eq:fsimple2}).

	\section{Example: Borwein Integral}\label{sec:app}

	There are two advantages of the new method. First, it allows one to evaluate certain integrals systematically without the ``clever tricks'' required to do the integrals in ordinary ways. Second, when the new methods are applicable they are usually much faster in comparison to ordinary methods. As an example, we now show how to use our new methods to understand the otherwise strange and curious behavior of the Borwein integrals in a straightforward way. 
	
Let us define the $n$th Borwein integral \cite{borwein2001some} as
	\begin{equation}B_n = \int_{-\infty}^\infty \prod_{k=1}^n\sinc\left(\frac{x}{2k-1}\right)\ dx= \int_{-\infty}^\infty \sinc\left(x\right)\sinc\left(\frac{x}{3}\right) \cdots \sinc\left(\frac{x}{2n-1}\right)\ dx.\end{equation}
	These integrals, first studied by David Borwein and Jonathan Borwein, begin as
	\begin{equation}B_1 = \pi,\ \ \ B_2 = \pi,\ \ \ B_3 = \pi,\ \ \ B_4 = \pi,\ \ \cdots\end{equation}
	A naive guess would be that the sequence of Borwein integrals is in fact constant at $\pi$ but, surprisingly, this is not so. The sequence begins to taper off at $8$th term, with $B_7 = \pi$ and 
	\begin{equation}B_8 = \pi - \frac{6879714958723010531}{467807924720320453655260875000}\pi \approx \left(1 - 1.47\times 10^{-11}\right)\pi.\end{equation}
	The strange behavior of this sequence, and its associated cousins, has attracted some attention in recent years, receiving a graphical proof \cite{schmid_two_2014}, as well some generalizations \cite{almkvist_more_2014}.

	To start with, let us apply equation (\ref{eq:delta0}) in the form of Proposition \ref{prop:funcdelta} to the $\sinc$ integral. We obtain
	\begin{eqnarray}\label{eq:sincst}\int_{-\infty}^\infty \sinc(x)\ dx &=& 2\pi \lim_{y \rightarrow 0} \sinc(-i\partial_y)\delta(y)\\
	&=& \pi\lim_{y \rightarrow 0}\frac{e^{\partial_y} - e^{-\partial_y}}{\partial_y}\delta(y)\\
	&=& \pi\lim_{y \rightarrow 0}\left(e^{\partial_y} - e^{-\partial_y}\right)\left(\Theta(y)+C\right)\\
	&=& \pi\lim_{y \rightarrow 0}\left[\left(\Theta(y+1)+C\right) - \left(\Theta(y-1)+C\right)\right]\\
	&=& \pi\left(\Theta(1) - \Theta(-1)\right)\\
	&=& \pi,\label{eq:sincend}
	\end{eqnarray}
	where in the second line we expressed the $\sinc$ function in terms of complex exponentials, in the third line we used the fact that the anti-derivative of the Dirac delta is the Heaviside $\Theta$ (with an undetermined constant $C$), and in the fourth line we used the fact that $e^{a\partial_y}$ acts formally as the translation operator $e^{a\partial_y}\varphi(y) = \varphi(x +a)$. Thus we see that equation (\ref{eq:delta0}) has integrated the $\sinc$ function over the real line, as claimed. We have explicitly kept the constant of integration here to show its eventual cancellation, as stated in Proposition \ref{prop:funcdelta}. From this point forward, we will suppress any further integration constants. 
	\par
	Now let us proceed onwards to the next Borwein integral. For notational simplicity, let us write the translation operator as $e^{a\partial} = T_a$. With the same procedure as before, we have
	\begin{align}
	\int_{-\infty}^\infty \sinc(x)\sinc\left(\frac{x}{3}\right) \ dx =& 2\pi \lim_{y \rightarrow 0} \sinc(-i\partial_y)\sinc\left(-i\frac{\partial_y}{3}\right)\delta(y)\\
	=& \frac{3\pi}{2}\lim_{y \rightarrow 0}\left(T_1-T_{-1}\right)\left(T_{\frac{1}{3}} - T_{-\frac{1}{3}}\right) \iint dy\ \delta(y)\\
	=& \frac{3\pi}{2}\lim_{y \rightarrow 0}\left(T_1-T_{-1}\right)\left(T_{\frac{1}3} - T_{-\frac{1}{3}}\right)R(y),
	\end{align}
	where the ramp function $R(y)$ is the anti-derivative of the Heaviside $\Theta$ function, i.e., $R(y) = y\Theta(y)$. Then we find
	\begin{equation}B_2=\frac{3\pi}{2}\left[R\left(1+\frac{1}{3}\right) + R\left(-1-\frac{1}{3}\right) - R\left(1-\frac{1}{3}\right) - R\left(-1+\frac{1}{3}\right)\right]. \end{equation}
	The ramp functions evaluated at positive arguments survive, and the rest go to zero. We are then left with
	\begin{equation}B_2=\frac{3\pi}{2}\left[R\left(1+\frac{1}{3}\right) - R\left(1-\frac{1}{3}\right)\right] = \frac{3\pi}{2}\left(1+\frac{1}{3} - 1 + \frac{1}{3}\right) = \pi,\end{equation}
	which is exactly what we expected.
	\par
	In general, with suitable notation, we can straightforwardly write down a general expression for the $n$th Borwein integral. Let $R_n(x)$ denote the $n$th polynomial ramp function as defined by equation (\ref{eq:polyramp}). Adopting the notation of Borwein [1], let $\gamma = (\gamma_1, \gamma_2, \cdots,\ \gamma_n)$ denote an $n$-tuple over $\gamma_i \in \{-1,1\}$ and denote $\mathrm{sign}(\gamma)$ as the product of all entries in $\gamma$. Finally, let 
	\begin{equation}\beta_\gamma = \sum_{k=1}^n \frac{\gamma_k}{2k-1}.\end{equation}
	Following the same procedure as before, we see that the $n$th Borwein integral will be expressed as a large sum over all possible combinations of shifts of different signs. With a bit of work, we can write down the integral in our notation as
	\begin{equation} B_n = \frac{(2n-1)!!\pi}{2^{n-1}}\sum_{\gamma \in \{-1,1\}^n} \mathrm{sign}(\gamma)R_{n-1}(\beta_\gamma),\end{equation}
	where the sum is over all $n$-tuples $\gamma \in \{-1,1\}^n$ and where the double factorial $(2n-1)!!$ denotes the product over odd numbers below $2n-1$, i.e. $(2n-1)!! = (2n-1)\times (2n-3) \times \cdots \times 3 \times 1$.
	\par
	In the sum above, we have $\beta_\gamma = -\beta_{-\gamma}$ so at most half the sum contributes due to the ramp function. Therefore, we can equivalently sum over the $n$-tuples $\gamma\in\{-1,1\}^n$ whose leading term $\gamma_1$ is $1$. We also have to introduce a sign correction for the $\gamma$s whose original leading term was $\gamma_1 = -1$. It follows we can write
	\begin{equation} B_n = \frac{(2n-1)!!\pi}{2^{n-1}(n-1)!}\sum_{\gamma;\gamma_1>0}\mathrm{sign}(\gamma)\mathrm{sign}(\beta_\gamma)\beta^{n-1}_\gamma,\end{equation}
	where the sums over $\gamma$ from this point on will be over $\gamma\in\{-1,1\}^n$ with leading term $\gamma_1 = 1$.
	\par
	Now we can distinguish between two cases. The first is when $\mathrm{sign}(\beta_\gamma) > 0$ for all $\gamma$. This happens if and only if $1 > \frac{1}{3} + \cdots + \frac{1}{2n-1}$. Using a trick from \cite{borwein2001some} we can perform the above sum explicitly in this case. Note that we have
	\begin{equation} \sum_{\gamma;\,\gamma_1>0}\mathrm{sign}(\gamma)e^{\beta_\gamma x} = e^{x} \prod_{k=2}^n \left(e^{\frac{x}{2k-1}} - e^{-\frac{x}{2k-1}}\right).\label{eq:trick}\end{equation}
	Since 
	\begin{equation}\left(e^{\frac{x}{2k-1}} - e^{-\frac{x}{2k-1}}\right) = \frac{2x}{2k-1} + \mathcal{O}(x^2),\end{equation}
	we can equate the coefficients of $x^{n-1}$ in the previous equation to get
	\begin{equation} \sum_{\gamma;\,\gamma_1>0}\mathrm{sign}(\gamma)\frac{\beta^{n-1}_\gamma}{(n-1)!} = \prod_{k=2}^n \frac{2}{2k-1} = \frac{2^{n-1}}{(2n-1)!!},\end{equation}
	and therefore we have
	\begin{equation}\sum_{\gamma;\,\gamma_1>0}\mathrm{sign}(\gamma)\beta^n_\gamma = \frac{2^{n-1}(n-1)!}{(2n-1)!!}.\end{equation}
	Of course, this means that we have $B_n = \pi$ so long as $1 > \frac{1}{3} + \cdots + \frac{1}{2n-1}$, which is satisfied for $n \le 7$ but not for $n > 7$. This is precisely the reason that the sequence begins to break off.
	\par
	In general, both the terms with $\mathrm{sign}(\beta_\gamma) > 0$ and $\mathrm{sign}(\beta_\gamma) < 0$ contribute to the sum and in this case we get  
	\begin{eqnarray}B_n &=& \frac{(2n-1)!!\pi}{2^{n-1}(n-1)!}\left(\sum_{\beta_\gamma > 0}\mathrm{sign}(\gamma)\beta^{n-1}_\gamma - \sum_{\beta_\gamma < 0}\mathrm{sign}(\gamma)\beta^{n-1}_\gamma\right)\\ 
	&=& \frac{(2n-1)!!\pi}{2^{n-1}(n-1)!}\left(\sum_{\gamma;\gamma_1>0}\mathrm{sign}(\gamma)\beta^{n-1}_\gamma - 2\sum_{\beta_\gamma < 0 }\mathrm{sign}(\gamma)\beta^{n-1}_\gamma\right)\\
	&=& \frac{(2n-1)!!\pi}{2^{n-1}(n-1)!}\left(\frac{2^{n-1}(n-1)!}{(2n-1)!!} - 2\sum_{\beta_\gamma < 0 }\mathrm{sign}(\gamma)\beta^{n-1}_\gamma\right)\\
	&=& \pi\left(1 - \frac{(2n-1)!!}{2^{n-2}(n-1)!} \sum_{\beta_\gamma < 0 }\mathrm{sign}(\gamma)\beta^{n-1}_\gamma\right),
	\end{eqnarray}
	where the latter sum over $\beta_\gamma < 0$ start contributing only when $n > 7$. Let us now look at the case $n=8$, where the latter sum involves only a single term given by
	\begin{equation*}\frac{(15)!!}{2^6\times 7!}\left(1-\frac{1}{3}-\frac{1}{5}-\frac{1}{7}-\frac{1}{9}-\frac{1}{11}-\frac{1}{13}-\frac{1}{15}\right)^7 = \frac{-6879714958723010531}{467807924720320453655260875000},\end{equation*}
	which means that $B_8$ has a deviation of the order $10^{-11}$ from $\pi$, as we saw previously.

	These are precisely the results originally obtained by the Borweins in \cite{borwein2001some}. But whereas the Borweins obtained their results through clever trigonometric manipulations, the form of the Borwein integrals fall out very naturally through equation (\ref{eq:delta0}). Following the same type of arguments, we can obtain the following general result for minimal additional effort.\\
	
	\begin{theorem}\cite{lord_91.40_2007}
		\rm: If $a_1, \cdots, a_m$, $b_1, \cdots, b_n$, and $c$ are positive real numbers such that 
		\begin{equation}c > a_1 + \cdots + a_m + b_1 + \cdots + b_n,\end{equation}
		then we have
		\begin{equation}I=\int_{-\infty}^\infty \sinc(a_1x)\cdots\sinc(a_mx)\cos(b_1x)\cdots\cos(b_nx)\,\sinc(cx)\ dx = \pi.\end{equation}
	\end{theorem}
	\begin{proof}
		\rm: First, we can without loss of generality assume that $c=1$ by dividing all our parameters by $c$. Expanding the sinc functions and the cosines in terms of complex exponentials, equation (\ref{eq:delta0}) allows us to write
		\begin{equation}I = \frac{\pi}{2^{n+m}}\left(\prod_{i=1}^m a_i\right)^{-1}\lim_{y\rightarrow 0}\prod_{i=1}^m\left(T_{a_i} - T_{-a_i}\right)\prod_{j=1}^n \left(T_{b_j} + T_{-b_j}\right)\left(T_1 - T_{-1}\right)R_m(y).\end{equation}
		Now let us define $\gamma$ similarly to before as an $(m+n)$-tuple over $\{1,-1\}$. We will define $\mathrm{sign}(\gamma)$ as the product of the first $m$ entries. Likewise, define
		\begin{equation}\beta_\gamma = \sum_{i=1}^m \gamma_i a_i + \sum_{j=1}^n \gamma_{j+m}b_j.\end{equation}
		Then we have
		\begin{equation}I = \frac{\pi}{2^{n+m}}\left(\prod_{i=1}^m a_i\right)^{-1}\sum_{\gamma\in \{1,-1\}^{n+m}}\mathrm{sign}(\gamma)\left[R_m(\beta_\gamma+1)-R_m(\beta_\gamma-1)\right].\end{equation}
		By assumption, we have
		\begin{equation}1 > a_1 + \cdots + a_m + b_1 + \cdots + b_n,\end{equation}
		so the term $R_m(\beta_\gamma - 1)$ must vanish for all $\gamma$ while $R_m(\beta_\gamma - 1)$ will be non-vanishing for all $\gamma$. The analog of equation (\ref{eq:trick}) that we need here is given by
		\begin{equation}\sum_{\gamma\in \{1,-1\}^{n+m}}\mathrm{sign}(\gamma)e^{(1+\beta_\gamma) x} = e^x\prod_{i=1}^m \left(e^{\frac{x}{a_i}} - e^{-\frac{x}{a_i}}\right)\prod_{j=1}^n \left(e^{\frac{x}{b_j}} - e^{-\frac{x}{b_j}}\right).\end{equation}
		Expanding and equating the coefficients of $x^m$, we must have
		\begin{equation}\sum_{\gamma\in \{1,-1\}^{n+m}}\mathrm{sign}(\gamma)\left(1+\beta_\gamma\right)^m = 2^{n+m}m!\prod_{i=1}^m a_i.\end{equation}
		It follows that we have
		\begin{eqnarray}I &=& \frac{\pi}{2^{n+m}}\left(\prod_{i=1}^m a_i\right)^{-1}\sum_{\gamma\in \{1,-1\}^{n+m}}\mathrm{sign}(\gamma)R_m(\beta_\gamma+1) \\
		&=& \frac{\pi}{2^{n+m}m!}\left(\prod_{i=1}^m a_i\right)^{-1}\sum_{\gamma\in \{1,-1\}^{n+m}}\mathrm{sign}(\gamma)(\beta_\gamma+1)^m\\
		&=&\frac{\pi}{2^{n+m}m!}\left(\prod_{i=1}^m a_i\right)^{-1}\left(2^{n+m}m!\prod_{i=1}^m a_i\right) = \pi,\end{eqnarray}
		where we've used the results of equation $(27)$ in the third equality. This is precisely the result as desired.
	\end{proof}
	
	Interestingly, the theorem above allows us to construct sequences of integrals which are constant at $\pi$ for arbitrary lengths before eventually suddenly tapering off, simply by choosing $c$ sufficiently large. If we choose $a_n$ and $b_n$ to be convergent as series, then we can have sequences of integrals which continue indefinitely to evaluate to $\pi$. \\

    Finally, let us point out that certain generalizations of the Borwein integrals are also naturally handled by our methods. Consider, for example, the inclusion of a Gaussian function $\exp(-x^2/2)$ into an integral of sincs, say
    \begin{equation}
    \int_{-\infty}^\infty \sinc^n(x)e^{-x^2/2}\ dx,
    \end{equation}
    with $n$ a positive integer. Such integrals are difficult to handle using traditional integration methods, and also computer algebra systems generally cannot evaluate such integrals (Maple and Mathematica cannot evaluate the above integrals for $n\ge 3$).
    
    With our methods however, the integrals are quite straightforward. Take for example the $n=3$ case, which is evaluated as
    \begin{eqnarray*}
    \int_{-\infty}^\infty\sinc^3(x)e^{-x^2/2}\ dx &= 2\pi \sinc^3(-i\partial_y)\exp(-(-i\partial_y)^2/2)\delta(y)\bigg|_{y=0}\\
    &= \sqrt{2\pi}\sinc^3(-i\partial_y)\exp(-y^2/2)\bigg|_{y=0}\\
    &= \frac{\sqrt{2\pi}}{8}(T_1 - T_{-1})^3 g(y)\bigg|_{y=0}\\
    &= \frac{\sqrt{2\pi}}{8}\left(g(3) - 3g(1) + 3g(-1) - g(-3)\right)\\
    &= \frac{-3(e^4-1)\sqrt{2\pi}+2\pi e^{9/2}(5\erf(3/\sqrt{2})-3\erf(1/\sqrt{2}))}{8e^{9/2}}\\
    &\approx 1.74815.
    \end{eqnarray*}
    In the above calculation $T_{\pm 1}$ denotes the translation operator by $\pm 1$, and $g$ denotes the second anti-derivative of the error function:
    \begin{equation}g(x) := D^{-3}\left(e^{-x^2/2}\right) = \frac{1}{2}xe^{-x^2/2}+\frac{1}{2}\sqrt{\frac{\pi}{2}}(1+x^2)\erf(x/\sqrt{2}).\end{equation}
    The new methods also allow one to straightforwardly solve all similar integrals that have extra factors such as a polynomial in $x$, an exponential in $x$ and/or arbitrary positive powers of trigonometric functions such as sine and cosine.
    
	\section{Proofs}\label{sec:proofs}
	In this section we prove the propositions listed in section \ref{sec:main}.
	
	\begin{proof}[Proof of Proposition \ref{prop:finint}]
		Let us denote the partial sums of $f$ by $f_N$. Then we have:
		\begin{eqnarray}
		\int_a^b f(x)e^{ixy}\ dx & = & \int_a^b \lim_{N\rightarrow \infty} f_N(x) e^{ixy}\ dx
		\\
		&=&\lim_{N\rightarrow \infty} \int_a^b f_N(x) e^{ixy}\ dx
		\\\label{eq:limoutint}
		&=&\lim_{N\rightarrow \infty} \int_a^b f_N(-i\partial_y) e^{ixy}\ dx
		\\
		&=&\lim_{N\rightarrow \infty} f_N(-i\partial_y ) \int_a^b  e^{ixy}\ dx
		\\\label{eq:difoutint}
		&=&\lim_{N\rightarrow \infty} f_N(-i\partial_y ) \frac{e^{iby }-e^{iay }}{iy } 
		\\
		&=&f(-i\partial_y ) \frac{e^{iby }-e^{iay }}{iy }.
		\end{eqnarray}
		The second line follows from the uniform convergence of a power series on a finite interval which allows us to exchange the order of integration and summation. The third line follows from the fact that $e^{ixy}$ is an eigenfunction of the operator $-i\partial_y$, with eigenvalue $x$. The fourth line is the Leibniz rule for differentiating under the integral sign, which is valid since $(-i\partial_y)^n e^{ixy}$ exists and is continuous for all finite $n$. The remaining steps are straightforward.
		
		The proof of equation (\ref{eq:ab}) follows analogously by replacing $e^{ixy}$ with $e^{xy}$.
	\end{proof}

	\begin{proof}[Proof of Proposition \ref{prop:intlap}]
		If the function $f$ is Laplace transformable, then by definition the limit
		\begin{equation} \lim_{a\rightarrow \infty}\int_0^a f(x)e^{-xy}\ dx = \lim_{a\rightarrow\infty}f\left(-\partial_y\right)\frac{1-e^{-ay}}{y}\end{equation}
		holds for all $y$ within the domain of convergence. The latter equality comes from an application of proposition \ref{prop:finint}, which holds here for functions $f$ which are entire. This establishes equation (\ref{eq:+}). Equation (\ref{eq:-}) follows analogously.
	\end{proof}

	\begin{proof}[Proof of Proposition \ref{prop:intfour}]
		If $f$ is Fourier transformable and entire, then the limit
		\begin{equation}\lim_{a\rightarrow \infty}\int_{-a}^a f(x)e^{ixy}\ dx = \lim_{a\rightarrow \infty}f(-i\partial_y)\frac{e^{iay}-e^{-iay}}{iy}\end{equation}
		exists by definition, where the latter equality follows from proposition \ref{prop:finint}. Thus we have
		\begin{equation}\int_{-\infty}^\infty f(x)e^{ixy}\ dx = \lim_{a\rightarrow \infty}f(-i\partial_y)\frac{e^{iay}-e^{-iay}}{iy} = \lim_{a\rightarrow\infty}f(-i\partial_y)\,2a\sinc{ay},\end{equation}
		as required.
	\end{proof}
	
	\begin{proof}[Proof of Proposition \ref{prop:intzerofreq}]
		The proofs of all the equations are essentially identical to the proofs of propositions \ref{prop:intlap} and \ref{prop:intfour}, but by taking the limit $y\rightarrow 0$ first. Note that equation (\ref{eq:intlap+-0}) is simply the sum of equations (\ref{eq:intlap+0}) and (\ref{eq:intlap-0}).
	\end{proof}

	\begin{proof}[Proof of Proposition \ref{prop:laplace}]
		Let us consider taking the Laplace transform of $f$ term-by-term. Then we get
		\begin{equation} \int_0^\infty f(x)e^{-xy}\ dx = \int_0^\infty \sum_{k=0}^\infty a_kx^ke^{-xy}\ dx.\end{equation}
		Each of the partial sums $\sum_{k=0}^Na_kx^ke^{-xy}$ is dominated by the function
		\begin{equation} \sum_{k=0}^\infty |a_k|x^ke^{-xy} = \tilde{f}(x)e^{-xy},\end{equation}
		which by assumption is integrable on $[0,\infty)$ for $y > y_0$. Applying the dominated convergence theorem, we may interchange summation and integration to get
		\begin{align} \int_0^\infty f(x)e^{-xy}\ dx &= \sum_{k=0}^\infty a_k\int_0^\infty x^ke^{-xy}\ dx\\
 		&= \sum_{k=0}^\infty a_k\frac{k!}{y^{k+1}}\\
		&= \sum_{k=0}^\infty a_k(-1)^k\frac{d^k}{dy^k}\frac{1}{y}\\
		&= f(-\partial_y)\ \frac{1}{y},\end{align}
		with the latter series convergent for $y > y_0$.

		Again, the proof of equation (\ref{eq:-}) follows analogously.
	\end{proof}

	Before proving Proposition \ref{prop:intlapsim}, we first prove a lemma showing that the power series operator $e^{b\partial}$ acts by translation on analytic functions.
	
	\begin{lemma}\label{lem:trans}
		Suppose function $f:\mathbb{R}\rightarrow\mathbb{R}$ is analytic on the open disk of radius $R$ around $x_0$ and $0<|b|<R$. Then
		\begin{equation}
		e^{b\partial_x} f(x_0)=f(x_0+b).
		\end{equation}
	\end{lemma}
	\begin{proof}
		Around $x_0$, $f$ has Taylor series
		\begin{align}\label{eq:tltr}
		f(x_0+b)=\sum_{n=0}^{\infty}\frac{1}{n!}f^{(n)}(x_0)~b^n.
		\end{align}
		On the other hand, by definition of $e^{b\partial_x}$,
		\begin{align}
		e^{b\partial_x}f(x)=\sum_{n=0}^{\infty}\frac{1}{n!} (b\partial_x)^n f(x)=\sum_{n=0}^{\infty}\frac{1}{n!}b^n f^{(n)}(x).
		\end{align}
		Evaluated at $x_0$, this gives the same result as (\ref{eq:tltr}).		
	\end{proof}

	\begin{proof}[Proof of Proposition \ref{prop:intlapsim}]
		The idea of the proof is to show that the formal action of $f(-\partial_x)$ in equation (\ref{eq:+0}) agrees with the results obtained by acting with $f(-\partial_x)$ as a power series operator in equation (\ref{eq:intlap+0}). Consider $f$ of the form
		\begin{equation}f(-\partial_x)=\sum_{j=1}^{N} c_j ~e^{b_j \partial_x}~(-\partial_x)^{n_j},\end{equation}
		with $b_i>0$, $c_j \in \mathbb{C}$, and $n_j\in\mathbb{Z}$. Letting $g(x) = (1-e^{-ax})/x$, we can write equation (\ref{eq:intlap+0}) as
		\begin{equation}\int_0^\infty f(x)\ dx = \lim_{a\rightarrow\infty}\lim_{x\rightarrow 0}\sum_{j=1}^N c_j ~e^{b_j \partial_x}~(-\partial_x)^{n_j}g(x).\end{equation}
		We will split the proof into two cases, corresponding to different signs of $n_j$.
		
		\vspace{0.6cm}
	    First, consider the terms with $n_j < 0$. We know that $f(-\partial_x)$ is a power series operator which acts on $g(x)$ as
		\begin{equation}f(-\partial_x)g(x) =\sum_{k=0}^\infty a_k(-\partial)^kg(x).\end{equation}
		
		Using the fact that integrating $n$ times and then differentiating $n$ times is the identity operator, we can write
		\begin{equation}g(x) = \partial_x^n\,\partial_x^{-n}g(x) = \partial_x^n g^{(-n)}(x),\end{equation}
		where $g^{(-n)}$ is any $n$th anti-derivative of $g$. Clearly the choice of anti-derivative will not matter here since there will always be a factor of $\partial_x^n$ acting on $g^{-(n)}$, so from now on we will choose an arbitrary representative. Thus, we can write
		\begin{equation}f(-\partial_x)g(x) = \sum_{k=0}^\infty a_k(-\partial_x)^k\partial^n_xg^{(-n)}(x) = \left[\partial_x^nf(-\partial_x)\right]g^{(-n)}(x).\end{equation}
		Applying this result to the terms with $n_j < 0$, we get
		\begin{align}\sum_{n_j<0} c_j ~e^{b_j \partial_x}~(-\partial_x)^{-|n_j|}g(x) &= \sum_{n_j<0} (-1)^{n_j}c_j ~e^{b_j \partial_x}g^{(-|n_j|)}(x)\\
		&=\sum_{n_j<0}(-1)^{n_j}c_jg^{(-|n_j|)}(x+b_j),\end{align}
		where the last equality follows acting with $e^{b\partial_x}$ as the translation operator, justified by lemma \ref{lem:trans} since all anti-derivatives of $g$ are entire. We can write $g^{-(n)}$ explicitly as
		\begin{equation}g^{(-n)}(x) = h^{(-n)}(x)+\frac{e^{-ax}}{a^{n-1}}p_{n-2}(x)-\frac{1}{(n-1)!}x^{n-1}\Ei(-ax),\end{equation}
		where $h^{(-n)}$ is an $n$th anti-derivative for $1/x$, $p_{n-2}(ax)$ is some polynomial of degree at most $n-2$ (we define $p_{-1}$ to be $0$), and $\Ei$ is the exponential integral defined by \begin{equation}\Ei(x):=-\int_{-x}^{\infty} \frac{e^{-y}}{y}~dy.\end{equation} 
		Therefore we have
		\begin{align}
		\sum_{n_j<0} (-1)^{n_j}c_jg^{(-|n_j|)}(x+b_j)=&\sum_{n_j<0}(-1)^{n_j}c_j\bigg(h^{(-|n_j|)}(x+b_j)+\frac{e^{-ax-ab_j}}{a^{|n_j|-1}}p_{|n_j|-2}(x+b_j)\nonumber\\&-\frac{x^{|n_j|-1}}{(|n_j|-1)!}\Ei(-ax-ab_j)\bigg).\end{align} 
		Taking the limits $x \rightarrow 0$ and $a\rightarrow \infty$ of the previous equation, we end up with
		\begin{align}
		&\lim_{a\rightarrow\infty}\lim_{x\rightarrow0}\sum_{n_j<0} (-1)^{n_j}c_jg^{(-|n_j|)}(x+b_j)\\
		=&\lim_{a\rightarrow\infty}\sum_{n_j<0}(-1)^{n_j}c_j\left(h^{(-|n_j|)}(b_j)+\frac{e^{-ab_j}}{a^{|n_j|-1}}p_{|n_j|-2}(b_j)\right)\\ 
		=& \sum_{n_j<0} (-1)^{n_j}c_j h^{(-|n_j|)}(b_j).\end{align}
		For $|n_j| > 1$, the term involving $\Ei$ vanishes as $x\rightarrow 0$ due to the presence of $x^{|n_j|-1}$. For $|n_j| = 1$, this term vanishes as $a\rightarrow \infty$ instead, with $\Ei(-ab_j) \rightarrow 0$ as $a\rightarrow \infty$.
		Note that this is precisely what is given by the formal action of $f(-\partial_x)$ on $1/x$.
		
		\vspace{0.6cm}
		Likewise, for the terms involving $n_j \ge 0$, we have
		\begin{equation}\sum_{n_j\ge0} c_j ~e^{b_j \partial_x}~(-\partial_x)^{n_j}g(x) = \sum_{n_j\ge0} (-1)^{n_j}c_j ~e^{b_j 										  
		\partial_x}g^{(n_j)}(x)=\sum_{n_j\ge0} (-1)^{n_j}c_jg^{(n_j)}(x+b_j).\end{equation}
		The $n$th derivative of $g$ is given by
		\begin{equation}g^{(n)}(x) = h^{(n)}(x) + \frac{e^{-ax}p_{n}(ax)}{x^{n+1}},\end{equation}
		where $h^{(n)}(x)$ is the $n$th derivative of $1/x$, and $p_n$ is an $n$th degree polynomial. Therefore we have
		\begin{align}&\lim_{a\rightarrow\infty}\lim_{x\rightarrow0}\sum_{n_j\ge0} (-1)^{n_j}c_j ~e^{b_j \partial_x}~(-\partial_x)^{n_j}g(x) \\
		=& \lim_{a\rightarrow\infty}\lim_{x\rightarrow0}\sum_{n_j\ge0} (-1)^{n_j}c_j \left(h^{(n)}(x+b_j)+\frac{e^{-ax-ab_j}p_{n}(ax+ab_j)}{(x+b_j)^{n+1}}\right) \\
        =&\lim_{a\rightarrow\infty}\sum_{n_j\ge0} (-1)^{n_j}c_j \left(h^{(n)}(b_j) + \frac{e^{-ab_j}p_{n}(ab_j)}{(b_j)^{n+1}}\right)\\
=&\sum_{n_j\ge0} (-1)^{n_j}c_j h^{(n)}(b_j),\end{align}
which again agrees with the formal action of $f(-\partial_x)$. Therefore equations (\ref{eq:+0}) and (\ref{eq:intlap+0}) agree in all cases, and this proves the desired result.
	\end{proof}

To prepare for the proofs of propositions \ref{prop:distdelta} and \ref{prop:funcdelta}, let us first begin with a brief review of distribution theory. An ordinary distribution is a continuous linear functional on the space $\mathcal{D}$ of smooth, compactly supported functions. The most familiar example of a distribution is of course the Dirac delta, which is just the evaluation functional at some point $y\in \mathbb{R}$:
\begin{equation}\langle \delta(x-y),\varphi(x)\rangle = \varphi(y) \equiv \int_\mathbb{R}\delta(x-y)\varphi(x)\ dx,\end{equation}
where the last expression is how the Dirac delta is usually denoted, but whose meaning is solely given in terms of the former equation. All locally integrable functions $f$ also define distributions as the kernel of an integral operator:
\begin{equation}\langle f,\varphi\rangle \equiv \int_\mathbb{R}f(x)\varphi(x)\ dx.\end{equation}
Distributions which can be written in integral form against a locally integrable function are called regular distributions. The Dirac delta is the canonical example of a singular, i.e. non-regular, distribution.

Distributions enjoy many operations, such as differentiation and multiplication by $x^n$. For equation (\ref{eq:delta}), an operation of relevance to us is the Fourier transform, which is defined for distributions by extending Parseval's theorem, i.e.
\begin{equation} \langle \mathcal{F}[f],\varphi\rangle = \langle f,\mathcal{F}[\varphi]\rangle. \label{eq:distfour}\end{equation}
However, note that the Fourier transform as defined above is not well defined on ordinary distributions. The reason is that the Fourier transform of a compactly supported function is not itself compactly supported. To introduce a consistent notion of Fourier transform, the space of test functions must be enlarged to be an invariant subspace of the Fourier transform, and such a space of functions is typically given by the space of Schwartz functions $\mathcal{S}$, i.e. the space of smooth functions where the function and all its derivatives decay faster than $x^n$ for all $n\ge 0$. The continuous dual space $\mathcal{S}'$ is then called the space of tempered distributions, and it can be shown that $\mathcal{S}'$ supports a well defined notion of Fourier transform given by equation (\ref{eq:distfour}).

It can be shown that the tempered distributional Fourier transform enjoys all the standard properties of the regular Fourier transform, such as the convolution theorem, and the conversion between multiplication by $x$ and differentiation. For example, the distributional Fourier transform of the tempered distribution $x^n$ is given by
\begin{equation}\mathcal{F}[x^n](y)=\sqrt{2\pi}(-i\partial_y)^n\delta(y).\end{equation}
Thus, given an entire function $f(x) = \sum_{k=0}^\infty a_kx^n$, if we are permitted to take the distributional Fourier transform term-by-term, then we may conclude that
\begin{equation}\hat{f}(y) = \sum_{k=0}^\infty a_k(-i\partial_y)^k\delta(y) \equiv \sqrt{2\pi}f(-i\partial_y)\delta(y)\label{eq:distseries},\end{equation}
from which equation (\ref{eq:delta}) would follow. The convergence of distributions is robust enough that the term-by-term evaluation of the Fourier transform is permitted, provided that the resulting series actually converges in the sense of distributions. Unfortunately, it can be shown that the distributional series
\begin{equation}\sum_{k=0}^\infty a_k(-i\partial_y)^k\delta(y)\end{equation}
actually diverges for every choice of $f$ possible. Thus, contrary to what might be expected at first glance, equation (\ref{eq:delta}) does not hold in the sense of distributions or even tempered distributions. The crux of the problem is that while test functions in $\mathcal{D}$ or $\mathcal{S}$ behave very well globally, they are not guaranteed to behave well locally. The terms $\partial^k_y\delta(y)$ probe the local structure of test functions while ignoring the global structure, and so the series becomes divergent\footnote{To establish the divergence formally, we use a result known as Borel's theorem \cite{Borel}. Borel's theorem says that given any sequence of $\{c_n\}_{n=0}^\infty$ of complex numbers, there exists a smooth, compactly supported function $\varphi$ such that $\partial^n\varphi(0) = c_n$. Therefore, no matter which analytic function we choose for $f$, we will always be able to find a test function such that the resulting series is divergent. This means that the series diverges in the sense of distributions for every function $f$.}. 

To rectify this problem, we must consider a space of test functions which are locally well-behaved. This naturally leads us to consider test functions which are bandlimited, i.e. test functions with compactly supported Fourier transforms. This problem was independently investigated by Gel'fand-Shilov \cite{Gelfand1,Gelfand2} and Ehrenpreis \cite{Ehrenpreis1,Ehrenpreis2}, and we summarize their results here.

As we did in section \ref{sec:main}, let $\mathcal{Z}$ denote the Fourier transform of $\mathcal{D}$, i.e. we denote by $\mathcal{Z}$ the space of functions whose Fourier transforms are smooth, compactly supported test functions. The Fourier transform of a compactly supported function can be holomorphically extended to the entire complex plane, and the class of functions which are the Fourier transforms of compactly supported functions, i.e. the space of bandlimited functions, is precisely characterized by the Paley-Wiener theorem.

\begin{theorem}[Paley-Wiener\cite{Rudin}]\label{prop:pw}
Let $F$ be the Fourier transform of a function $f$ which is compactly supported on the interval $[-A,A] \subseteq \mathbb{R}$. Then $F$ is entire on $\mathbb{C}$ and square integrable on $\mathbb{R}$. Moreover, there exists some constant $C>0$ such that $F$ satisfies the inequality
\begin{equation}|F(z)| \le Ce^{A|\mathrm{Im}(z)|}\label{eq:expineq}.\end{equation}
Conversely, any entire function $F$ which is square integrable on $\mathbb{R}$ and satisfies (\ref{eq:expineq}) is the Fourier transform of some function $f$ which is compactly supported on $[-A,A]$.
\end{theorem}

The classical Paley-Wiener theorem assumes no smoothness conditions on the function $f$. To describe the space $\mathcal{Z}$ we must extend the Paley-Weiner theorem to characterize the Fourier transforms of smooth, compactly supported functions. This was done by Gelf'and and Shilov, as given in the following theorem.

\begin{theorem}[Gel'fand - Shilov\cite{Gelfand1}]\label{prop:gs}
Let $F\in \mathcal{Z}$ be the Fourier transform of a function $f \in \mathcal{D}$ which is compactly supported on $[-A,A] \subseteq \mathbb{R}$. Then $F$ is entire, and for each $n\in \mathbb{N}$, there exists some constant $C_n > 0$ such that
\begin{equation}|z^nF(z)| \le C_ne^{A|\mathrm{Im}(z)|}\label{eq:expineq2}.\end{equation} 
Conversely, any entire function $F$ which satisfies the family of inequalities (\ref{eq:expineq2}) is the Fourier transform of a smooth test function $f \in \mathcal{D}$ which is compactly supported on $[-A,A]$.
\end{theorem}

In particular, the Fourier transform of a smooth compactly supported function is a Schwartz function, so that $F$ is Schwartz on $\mathbb{R}$. In fact, the space $\mathcal{Z}$ is a proper subspace of the Schwartz space $\mathcal{S}$ which is dense in $L^p(\mathbb{R})$ for $1 \le p < \infty$. The topology of $\mathcal{Z}$ will be defined so as to make the Fourier transform a topological isomorphism, i.e. a sequence of functions $F_n \in \mathcal{Z}$ converges to $F\in \mathcal{Z}$ if and only if the corresponding Fourier transform $f_n \in \mathcal{D}$ converges to $f \in \mathcal{D}$. The continuous dual space to $\mathcal{Z}$ will then be denoted $\mathcal{Z}'$. There does not seem to be a standard name for the spaces $\mathcal{Z}$ and $\mathcal{Z}'$ in the literature, so we will call $\mathcal{Z}$ the space of Paley-Wiener (PW) test functions, and we will call $\mathcal{Z}'$ the space of Paley-Wiener (PW) distributions. 

Since the Fourier transform $\mathcal{F}:\mathcal{D}\rightarrow \mathcal{Z}$ is an isomorphism between $\mathcal{D}$ and $\mathcal{Z}$, it follows that it also induces an isomorphism between the dual spaces $\mathcal{D}'$ and $\mathcal{Z}'$. The Fourier transform of a distribution $u \in \mathcal{D}'$ will be a PW distribution $\hat{u}\in \mathcal{Z}'$ defined by
\begin{equation}\langle \hat{u},\varphi \rangle = \langle u, \hat\varphi\rangle\end{equation}
for $\varphi \in \mathcal{Z}$ and $\hat\varphi \in \mathcal{D}$. Note that since $\mathcal{F}$ is a topological isomorphism, if we have a distributional series $\sum_{k=0}^\infty u_n \in \mathcal{D}'$ which converges to $u \in \mathcal{D}'$, then we may take the Fourier transform term-by-term to get $\sum_{k=0}^\infty\hat{u}_n$, and the result converges to $\hat u \in \mathcal{Z}'$. Thus we see that the space $\mathcal{Z}'$ is the proper setting to formulate equation (\ref{eq:delta}).

\begin{proof}[Proof of Proposition \ref{prop:distdelta}]
Since $f\in L^1(\mathbb{R}) + L^2(\mathbb{R})$, it follows that $f$ defines a regular distribution in $\mathcal{D}'$. Since $f$ is also entire, its power series converges absolutely and uniformly on every compact subset of $\mathbb{R}$, and so the partial sums of $f$ converges distributionally to $f$ in $\mathcal{D}'$. Since the partial sums converge, we are permitted to take the distributional Fourier transform term-by-term to conclude that
\begin{equation}\mathcal{F}[f](y) = \mathcal{F}\left[\sum_{k=0}^\infty a_kx^k\right] = \sum_{k=0}^\infty a_k\mathcal{F}[x^k](y) = \sqrt{2\pi}\sum_{k=0}^\infty a_k(-i\partial_y)^k\delta(y),\end{equation}
which holds as an equality of PW distributions in $\mathcal{Z}'$. In the above equation, we've used the fact that the distributional Fourier transform of $x^n$ is given by $\sqrt{2\pi}(-i\partial_y)^k\delta(y)$.
\end{proof}

In fact, there is no general need to restrict $f$ to be in $L^1(\mathbb{R})+L^2(\mathbb{R})$. The only difference then is that $\hat{f}$ will in general not be a regular distribution.

\begin{proof}[Proof of Proposition \ref{prop:funcdelta}] The operator $f(-i\partial_x)$ for $f(x)$ of the form (\ref{eq:fsimple2}) can be written as
\begin{equation}f(-i\partial_x)= \sum_{k=1}^{N}\Big[c_{k}e^{b_{k}\partial_{x}}+d_{k}e^{-b_{k}\partial_{x}}\Big]\partial_{x}^{-n_{k}}.\end{equation}
First, let us calculate the right-hand side of (\ref{eq:fourier_delta}) formally using $f(-i\partial_x)$ given above. The $n$th anti-derivative of the Dirac delta distribution is given by
\begin{equation}
\partial^{-n}_x\delta(x) = R_{n-1}(x) = \frac{x^{n-1}}{(n-1)!}\Theta(x).\end{equation}
The exponentials then act formally by translation, giving us
\begin{align}f\left(-i\partial_x\right)\delta(x)&=\sum_{k=1}^{N}\Big[c_{k}e^{b_{k}\partial_{x}}+d_{k}e^{-b_{k}\partial_{x}}\Big]\partial_{x}^{-n_{k}}\delta(x)\\
&=\sum_{k=1}^{N}\Big[c_{k}e^{b_{k}\partial_{x}}+d_{k}e^{-b_{k}\partial_{x}}\Big]R_{n_{k}}(x)\\
&=\sum_{k=1}^{N}\Big[c_{k}R_{n_{k}}(x+b_{k})+d_{k}R_{n_{k}}(x-b_{k})\Big].\end{align}
On the other hand, by Proposition \ref{prop:distdelta}, we have distributionally 
\begin{eqnarray}\hat{f}\left(x+y\right)&=&\sqrt{2\pi}\sum_{k=1}^{N}\Big[c_{k}e^{b_{k}\partial_{x}}+d_{k}e^{-b_{k}\partial_{x}}\Big]\partial^{-n_k}_x\delta(x+y)\\
&=&\sqrt{2\pi}\sum_{k=1}^{N}\Big[c_{k}e^{b_{k}\partial_{x}}+d_{k}e^{-b_{k}\partial_{x}}\Big]R_{n_{k}-1}(x+y),\end{eqnarray}
where we have  used the fact that $R_{n-1}(x)$ is an $n$th distributional anti-derivative of $\delta(x)$, i.e., $\partial^n_xR_{n-1}(x) = \delta(x)$. With the same reasoning given before in proposition \ref{prop:intlapsim}, we may choose any representative for the anti-derivative since there will always be a term involving $\partial^n_x$ which acts on $R_{n-1}(x)$. Namely, the identity continues to hold with $R_{n-1}(x) + p_{n-1}(x)$ for any polynomial $p_{n-1}(x)$ of degree at most $n-1$. For simplicity, we continue by taking $p_{n-1}=0$. Applying this identity to an arbitrary test function $\varphi \in \mathcal{Z}$, we get
\begin{align}\int_{\mathbb{R}}\hat{f}(x+y)\varphi(x)\ dx&=\sqrt{2\pi}\sum_{k=1}^{N}\int_{\mathbb{R}}R_{n_{k}-1}(x+y)\Big[c_{k}e^{-b_{k}\partial_{x}}+d_{k}e^{b_{k}\partial_{x}}\Big]\varphi(x)\ dx\\
&=\sqrt{2\pi}\sum_{k=1}^{N}\int_{\mathbb{R}}R_{n_{k}-1}(x+y)\Big[c_{k}\varphi\left(x-b_{k}\right)+d_{k}\varphi\left(x+b_{k}\right)\Big]\ dx.\end{align}
This holds for all PW test functions, so we may take a sequence $\varphi_n$ which converges weakly to the Dirac delta to conclude that
\begin{equation}\hat{f}\left(y\right)=\sqrt{2\pi}\sum_{k=1}^{N}\left[c_{k}R_{n_{k}}\left(y+b_{k}\right)+d_{k}R_{n_{k}}\left(y-b_{k}\right)\right],\end{equation}
where the above equation holds as an equality of functions. Comparing this expression with the result obtained earlier by acting with $f(-i\partial_x)$ formally, we therefore have
\begin{equation} \hat{f}(y) = \sqrt{2\pi}f(-i\partial_y)\delta(y),\end{equation}
which is precisely equation (\ref{eq:delta}). Taking the limit $y\rightarrow 0$ yields equation (\ref{eq:delta0}).
\end{proof}

\section{Conclusions and Outlook}

In this paper we have stated and proven several results concerning the applicability of integration by differentiation formulas first stated in \cite{kempf2014new,kempf2015path}.

First, Proposition \ref{prop:finint} establishes equation (\ref{eq:abi}) as the key result for integration on finite intervals. As long as the integrand $f$ has a convergent power series covering the interval of integration $[a,b]$, the formula will apply.

Taking the limit of equation (\ref{eq:abi}) to the case of the half-line or entire real line gives us Propositions \ref{prop:intlap}, \ref{prop:intzerofreq}, and \ref{prop:intfour}. These propositions hold for entire functions as long as the respective integrals are convergent.

We found that there are nontrivial issues regarding the domain of convergence for equations (\ref{eq:+}) and (\ref{eq:-}). The domain of convergence for the equations is generally smaller than the full domain of convergence for the Laplace transform. These solutions need analytic continuation. In Proposition \ref{prop:intlap}, however, we gave  regularized versions of these formulas which are guaranteed to be valid on the full domain of convergence for the Laplace transform. 
Further, Proposition \ref{prop:laplace} gives simple conditions for 
(\ref{eq:+}) and (\ref{eq:-}) to hold. 

The same considerations hold for equations (\ref{eq:-0}) - (\ref{eq:+-0}), which are the zero frequency cases of equations (\ref{eq:+}) and (\ref{eq:-}) used for integration. Since $y=0$ is in general not a limit point of the domain of convergence, the domain of validity of equations (\ref{eq:+}) and (\ref{eq:-}) is nontrivial. Useful, therefore, for practical purposes is Proposition \ref{prop:intlapsim} which establishes that for entire, integrable functions of the form
\begin{align}
	f(x)=\sum_{j=1}^{N} c_j e^{-b_j x} x^{n_j},
\end{align}
with $b_i>0$, $c_j \in \mathbb{C}$, and $n_j\in\mathbb{Z}$, the equations (\ref{eq:+0}) - (\ref{eq:+-0}) are valid with $e^{a\partial_y}$ acting formally as the translation operator by $a$, and $\partial_y^n$ acting as repeated differentiation $(n>0)$ or integration $(n<0)$. Any undetermined integration constants arising will ultimately cancel and can, therefore, be neglected.

Finally, in Proposition \ref{prop:distdelta} we found that equation (\ref{eq:delta}) holds as a distributional identity defined on the space of Paley-Wiener test functions, i.e., the space of functions whose Fourier transforms are smooth, compactly supported test functions. If $f \in L^1(\mathbb{R}) + L^2(\mathbb{R})$ is an entire function of the form
\begin{equation}f(x) = \sum_{n=1}^N\frac{c_ke^{ib_kx} + d_ke^{-ib_kx}}{x^{n_k}},\end{equation}
for constants $c_k,d_k\in\mathbb{C}$, $b_k \in \mathbb{R}$, and $n_k \in \mathbb{N}_+$, then Proposition \ref{prop:funcdelta} says that equation (\ref{eq:delta}) can be regarded to hold in the usual (non-distributional) sense. Here, the exponential and the derivative operators act formally by translation and differentiation/integration, with the $n$th anti-derivative of the Dirac delta being given by generalized ramp function
\begin{equation}\partial^{-n}_x\delta(x) = R_{n-1}(x) = \frac{x^{n-1}}{(n-1)!}\Theta(x).\end{equation}

\vspace{0.6cm}
Regarding concrete applications of the integration by differentiation techniques, our example of the Borwein integrals clearly shows that the curious behavior of the Borwein sequence comes from the breaking of symmetry in the associated translations used in the evaluation of the integrals. Traditionally, most of the Borwein style integrals involving products of cosines and sincs have been established through Fourier analytic methods, and in light of equation (\ref{eq:delta}), it is perhaps unsurprising that the integration through differentiation technique would be so naturally suited to their evaluation.

Going forward, it should be interesting to expand the toolbox of our integration by differentiation methods by making use of Green's function techniques. Recall, for example, the way we use the integration by differentiation methods to evaluate $\int \sin(x)/x\ dx$ in two steps from (\ref{eq:sincst}) to (\ref{eq:sincend}). We first calculate $[1/(-\partial_y)]\delta(y)$, which, crucially, is the Green's function of $-\partial_y$. We can then straightforwardly apply $\sin(-\partial_y)$ on the result before multiplying by $2\pi$ and taking the limit of $y\rightarrow 0$. 

More generally, for integrals of the form $\int g(x)/f(x)\ dx$ (or for the corresponding Fourier transforms), 
equation (\ref{eq:delta0}) yields: 
\begin{equation}\int_{-\infty}^\infty\frac{g(x)}{f(x)}\ dx = 2\pi \lim_{y\rightarrow 0} \frac{g(-i\partial_y)}{f(-i\partial_y)}\delta(y).\end{equation} 
\noindent If we can solve the corresponding Green's function problem $
f(-i\partial_y)G(y,z)=\delta(y-z)$, i.e.,
\begin{equation}G(y,z) = \frac{1}{f(-i\partial_y)}\delta(y-z)\end{equation}
for $G(y,z)$, then the integral can be obtained from 
\begin{equation}2\pi \lim_{y\rightarrow0}g(-i\partial_y)G(y,0).\end{equation} 
Namely, assume that the integral of a function $g(x)/f(x)$  is difficult to obtain directly, but that the Green's function for $f(-i\partial_y)$ can be found. As the further actions of  $g(-i\partial_y)$ and the limit taking are straightforward, the result of the integral (or the Fourier transform) can then be obtained easily. 

For example, consider the integral with $f(x)=x^2+1$ and $g(x)=\cos(x)$:
\begin{equation}\int_{-\infty}^\infty \frac{\cos(x)}{x^2+1}\ dx\end{equation}
The integrand is a non-analytic function on the real line. Nevertheless, we may proceed heuristically as follows. Starting with equation (\ref{eq:delta}), we have
\begin{equation}\int_{-\infty}^\infty \frac{\cos(x)}{x^2+1}= 2\pi \lim_{y\rightarrow 0}\cos(-i\partial_y)(-\partial_y^2+1)^{-1}\delta(y)\end{equation}
We recognize the term $(-\partial_y^2+1)^{-1}\delta(y)$ as the Green's function for the differential operator $(-\partial_y^2+1)$, which is known to be 
\begin{equation}G(y) = \frac{1}{2}e^{-|y|}.\end{equation}
Notice in particular the non-analyticity at $y=0$ which would render the power series method inapplicable. Continuing, our methods directly yield: 
\begin{eqnarray}\int_{-\infty}^\infty \frac{\cos(x)}{x^2+1} &=& \pi \lim_{y\rightarrow 0}\cos(-i\partial_y)e^{-|y|}\\ &=&\frac{\pi}{2}\lim_{y\rightarrow 0}(e^{\partial_y} + e^{-\partial_y})e^{-|y|}\\
&=&\frac{\pi}{2}\lim_{y\rightarrow 0}\left(e^{-|y+1|}+e^{-|y-1|}\right)\\
&=&\frac{\pi}{e}.\end{eqnarray}
So far, all of our methods have focused on practical applications of the integration by differentiation formulas. This has motivated us to define the operator $f(r\partial_x)$ in terms of a convergent power series. However, as we briefly mentioned at the beginning of Sec.\ref{sec:main}, the operator $f(r\partial_x)$ can be given a much broader definition using the functional calculus resulting from the spectral theorem. Namely, we may define the operator as
\begin{equation} f(r\partial_x) = \mathcal{F}^{-1}M_{f(-irx)}\mathcal{F},\end{equation}
where $\mathcal{F}$ is the Fourier transform, and where $M_{f(-irx)}$ is the multiplication operator by $f(-irx)$. Note however, that the evaluation of the operator using the spectral definition requires prior knowledge of the Fourier transform and therefore also of the integral as the zero frequency limit. 
The spectral calculus definition is, therefore, of limited practical use for the evaluation of integrals. 
On the other hand, the broader validity of the spectral calculus based definition will allow us to extend the validity of the integration by differentiation formulas to a much greater class of functions, namely non-analytic functions. Such an extension would allow all of our equations to hold as distributional equalities for a large class of admissible functions. This may prove useful in the theoretical studies of quantum systems, where such distributional methods are widely applied. 

\vspace{0.6cm}
Finally, let us recall that these integration by differentiation methods were originally motivated by the challenges of the Feynman path integral. A suitable generalization of our present results may make the path integral better defined and/or more straightforward to evaluate in practice. Quantum field theory is plagued by ultraviolet divergences and rigorous mathematical approaches to quantum field theory commonly define quantum fields as operator-valued tempered distributions. However, as we saw earlier, the space of Paley-Wiener distributions is a space with a naturally built in ultraviolet cutoff. It would be interesting to see if the space of Paley-Wiener distributions has a role to play in the mathematical foundations of quantum field theory. 
$$$$
\noindent \bf Acknowledgements: \rm AK acknowledges support through the Discovery Program of the National Science and Engineering Research Council (NSERC) of Canada. DJ acknowledges useful discussions with Yidong Chen.

\section*{References}

\bibliographystyle{unsrt}
\bibliography{ibd}

\end{document}